\pdfoutput=1
\ifx\synctex\undefined\else\synctex=1\fi
\documentclass[a4paper,UKenglish]{lipics-v2018}

\usepackage{microtype}%

\usepackage{bm}
\usepackage{stmaryrd}
\usepackage{mathtools}
\usepackage{complexity}
\usepackage{xspace}

\newcommand{\defeq}{\stackrel{\scriptscriptstyle\text{def}}{=}}

\newcommand{\etal}{\text{et al.}\xspace} %
\newcommand{\ra}{\rightarrow}

\newcommand{\N}{\mathbb{N}}                    %
\renewcommand{\vec}[1]{\bm{#1}}                %
\newcommand{\set}[1]{\left\{#1\right\}}        %
\newcommand{\tuple}[1]{[#1]}                 
\renewcommand{\vec}[1]{\bm{#1}}                %
\newcommand{\pop}[1]{\mathrm{Pop}(#1)}         %
\renewcommand{\PP}{\mathcal{P}}                  %
\newcommand{\trans}[1]{\xrightarrow{#1}}       %
\newcommand{\pre}{\mathit{pre}} %
\newcommand{\post}{\mathit{post}} %
\newcommand{\postco}{\mathit{post}_a}
\newcommand{\preco}{\mathit{pre}_a}

\newcommand{\unorm}[1]{\|{#1}\|_u}
\newcommand{\lnorm}[1]{\|{#1}\|_l}
\newcommand{\sem}[1]{\llbracket{#1}\rrbracket}
\newcommand{\fire}[1]{{#1}_{t^*}}

\newcommand{\ecov}[1]{\geq_{#1}}

\theoremstyle{plain}
\newtheorem{proposition}[theorem]{Proposition} %

\theoremstyle{remark}
\newtheorem{remarki}[theorem]{Remark} %

\title{Verification of Immediate Observation Population Protocols}
\author{Javier Esparza}{Technische Universität München, Munich, Germany}{esparza@in.tum.de}{0000-0001-9862-4919}{Supported by ERC Advanced Grant (787367: PaVeS).}
\author{Pierre Ganty}{IMDEA Software Institute, Madrid, Spain}{pierre.ganty@imdea.org}{0000-0002-3625-6003}{Supported by Madrid Regional Government project S2013/ICE-2731, N-Greens Software - Next-GeneRation Energy-EfficieNt Secure Software, the Spanish Ministry of Economy and Competitiveness project No. TIN2015-71819-P, RISCO - RIgorous analysis of Sophisticated COncurrent and distributed systems, and by a Ramón y Cajal fellowship RYC-2016-20281.}
\author{Rupak Majumdar}{MPI-SWS, Kaiserslautern, Germany}{rupak@mpi-sws.org}{}{supported by the  ERC Synergy award (IMPACT).}
\author{Chana Weil-Kennedy}{Technische Universität München, Munich, Germany}{chana.wk@gmail.com}{}{Part of this work was done during a visit at the IMDEA Software Institute.}
\authorrunning{J. Esparza and P. Ganty and R. Majumdar and C. Weil-Kennedy} %

\Copyright{J. Esparza and P. Ganty and R. Majumdar and C. Weil-Kennedy}%

\subjclass{Theory of Computation \textrightarrow{} Models of Computation \textrightarrow{} Concurrency \textrightarrow{} Distributed Computing models}%


\EventEditors{Sven Schewe and Lijun Zhang}
\EventNoEds{2}
\EventLongTitle{29th International Conference on Concurrency Theory (CONCUR 2018)}
\EventShortTitle{CONCUR 2018}
\EventAcronym{CONCUR}
\EventYear{2018}
\EventDate{September 4--7, 2018}
\EventLocation{Beijing, China}
\EventLogo{}
\SeriesVolume{118}
\ArticleNo{31} %
\nolinenumbers %
\hideLIPIcs  %

\begin{document}

\maketitle

\begin{abstract}
Population protocols (Angluin et al., \textit{PODC}, 2004) are a formal model of sensor
networks consisting of identical mobile devices. Two devices can interact 
and thereby change their states. Computations are
infinite sequences of interactions satisfying a strong fairness constraint.

A population protocol is well-specified if for every initial configuration $C$
of devices, and every computation starting at $C$, all devices eventually agree
on a consensus value depending only on $C$. If a protocol is well-specified,
then it is said to compute the predicate that assigns to each initial
configuration its consensus value.

In a previous paper we have shown that the problem whether a given protocol is well-specified
and the problem whether it computes a given predicate are decidable. 
However, in the same paper we prove that both problems are at least as hard as the reachability problem for Petri nets. 
Since all known algorithms for Petri net reachability have non-primitive recursive complexity, in this 
paper we restrict attention to immediate observation (IO) population protocols, a class introduced and studied in (Angluin et al., \textit{PODC}, 2006). 
We show that both problems are solvable in exponential space for IO protocols. 
This is the first syntactically defined, interesting class of protocols for which an algorithm 
not requiring Petri net reachability is found. 
\end{abstract}

\section{Introduction}

Population protocols~\cite{DBLP:conf/podc/AngluinADFP04,DBLP:journals/dc/AngluinADFP06}
are a model of distributed, concurrent computation by anonymous, identical finite-state agents. 
They capture the essence of distributed computation in different areas. 
In particular, even though they were introduced to  model networks of passively mobile sensors, 
they are also being studied in the context of natural computing~\cite{Navlakha:2014:DIP:2688498.2678280,DBLP:conf/sss/ChatzigiannakisMS10}. 
They also exhibit many common features with Petri nets, another fundamental model of concurrency.

A protocol has a finite set of states $Q$ and a set of transitions of the form \((q, q') \mapsto (r, r')\), where \(q, q', r, r' \in Q\). 
If two agents are in states, say, \(q_1\) and \(q_2\), and the protocol has a transition of the form \((q_1, q_2) \mapsto (q_3, q_4)\), 
then the agents can interact and simultaneously move to states \(q_3\) and \(q_4\). 
Since agents are anonymous and identical, the global state of a protocol is completely determined by the number of agents at each local state, called a \emph{configuration}. 
A protocol \emph{computes} a boolean value for a given initial configuration if in all fair executions starting at it, all agents eventually agree to this value%
\footnote{An execution is fair if it is finite and cannot be extended, or it is infinite and satisfies the following condition: if \(C\) appears infinitely often in the execution, then every step enabled at \(C\) is taken infinitely often in the execution.}%
---so, intuitively, population protocols compute by reaching a stable consensus. 
Observe that a protocol may compute no value for some initial configuration, in which case it is deemed not \emph{well-specified}~\cite{DBLP:conf/podc/AngluinADFP04}.
 
Population protocols are parameterized systems. 
Every initial configuration yields a different finite-state instance of the protocol, and the specification is a 
global property of the infinite family of protocol instances so generated. 
More precisely, the specification is a predicate $P(x)$ stipulating the boolean value $P(C)$ that the protocol must compute from the initial configuration $C$. 

Initial verification efforts for verifying population protocols studied the problem of checking if $P(x)$ is correctly computed for 
a \emph{finite} set of initial configurations, a task within the reach of finite-state model checkers. 
In 2015 we obtained the first positive result on parameterized verification \cite{DBLP:conf/concur/EsparzaGLM15}. 
We showed that the problem of deciding if a given protocol is well-specified for all initial configurations is decidable. 
The same result holds for the correctness problem: given a protocol and a predicate, deciding if the protocol is well-specified and computes the predicate. 
Unfortunately, we also showed \cite{DBLP:conf/concur/EsparzaGLM15,esparza_et_al:LIPIcs:2016:6862} that both problems are as hard as the reachability problem for Petri nets. 
Since all known algorithms for Petri net reachability run in non-primitive recursive time in the worst case, the applicability of this result is limited. 

In this paper we initiate the investigation of subclasses of protocols with a more tractable well specification and correctness problems. 
We focus on the subclass of \emph{immediate observation} protocols (IO protocols), introduced and studied by
Angluin \etal{} \cite{DBLP:journals/dc/AngluinAER07}. 
These are protocols whose transitions have the form $(q_1, q_2) \mapsto (q_1, q_3)$. 
Intuitively, in an IO protocol an agent can change its state from $q_2$ to $q_3$ by \emph{observing} that another agent is in state $q_1$. 
This yields an elegant model of protocols in which agents interact through \emph{sensing}: 
If an agent in state $q_2$ senses the presence of  another agent in state $q_1$, then it can change its state to $q_3$. 
The other agent typically does not even know that it has been sensed, and so it keeps its current state. 
They also capture the notion of catalysts in chemical reaction networks.

Angluin \etal{} focused on the expressive power of IO protocols.
Our main result is that for IO protocols, both the well specification and correctness problems can be solved in \EXPSPACE{} (we also show the problem
is \PSPACE-hard). 
This is the first time that the verification problems of a substantial class of protocols are proved to be solvable in elementary time.
To ensure elementary time, our proof uses techniques significantly different from previous results \cite{DBLP:conf/concur/EsparzaGLM15}.
The key to our result is the use of \emph{counting constraints} to symbolically represent possibly infinite (but not necessarily upward-closed)
sets of configurations.  
A counting constraint is a boolean combination of atomic threshold constraints of the form $x_i \geq k$. 
We prove that, contrary to the case of arbitrary protocols, the set of configurations reachable from a counting set 
(the set of solutions of a counting constraint) is again a counting set and we characterize the complexity of 
representing this set.
We believe that this result can be of independent interest for other parameterized systems.

Angluin \etal \cite{DBLP:journals/dc/AngluinAER07} proved that 
IO protocols compute exactly the predicates represented by counting constraints.  
Our main theorem yields a new proof of this result as a corollary. 
But it also goes further. 
Using our complexity results, we can provide a lower bound on the state complexity of IO protocols, 
i.e., on the number of states necessary to compute a given predicate. 
These results complement recent bounds obtained for arbitrary protocols \cite{blondin_et_al:LIPIcs:2018:8511}.

\section{Immediate Observation Population Protocols}%
\label{sec:preliminaries}
\subsection{Preliminaries} 

A \emph{multiset} on a finite set \(E\) is a mapping \(C \colon E \rightarrow \N\), 
thus, for any $e\in E$,
\(C(e)\) denotes the number of occurrences of element \(e\) in \(C\). 
Operations on \(\N\) like addition, subtraction, or comparison, are extended to multisets by
defining them component wise on each element of \(E\). 
Given $e \in E$, we denote by $\vec{e}$ the multiset consisting of one occurrence of element
$e$, that is, the multiset satisfying $\vec{e}(e)=1$ and 
$\vec{e}(e')=0$ for every $e' \neq e$. 
Given \(E'\subseteq E\) define \(C(E')\defeq\sum_{e\in E'} C(e)\).
Given a total order $e_1 \prec e_2 \prec \cdots \prec e_n$ on $E$, a multiset $C$ can be 
equivalently represented by the vector $(C(e_1), \ldots, C(e_n))\in \N^n$.

\subsection{Protocol Schemes}

A \emph{protocol scheme} $\mathcal{A} = (Q, \Delta)$ consists of a finite 
non-empty set $Q$ of states and a set $\Delta \subseteq Q^4$.
If $(q_1,q_2, q_1',q_2')\in \Delta$, we write $(q_1,q_2) \mapsto (q_1',q_2')$ and call it 
a \emph{transition}. 

\emph{Confugurations} of a protocol scheme $\mathcal{A}$ are given by \emph{populations}.
A population \(P\) is a multiset on \(Q\) with at least two elements, i.e., \(P(Q) \geq 2\).
The set of all populations is denoted $\pop{Q}$.
Intuitively, a configuration $C\in \pop{Q}$ describes a collection of identical finite-state {\em agents} 
with $Q$ as set of states, containing $C(q)$ agents in state $q$.

Pairs  of agents \emph{interact} using transitions from $\Delta$.
Formally, given two configurations $C$ and $C'$ and a transition $\delta =(q_1,q_2)\mapsto (q_1',q_2')$, we write $C \xrightarrow{\delta} C'$ if 
\[C \geq (\vec{q}_1 + \vec{q}_2) \text{ holds, and } C' = C - (\vec{q}_1 + \vec{q}_2) + (\vec{q}'_1 + \vec{q}'_2) \enspace .\]
(Recall that $\vec{q}$ is the multiset consisting only of one occurrence of $q$.)
From the definition of interaction, it is easily seen that, inside the tuple \( (q_1,q_2,q_1',q_2') \in \Delta \), the ordering between \(q_1\) and \(q_2\) and between \(q_1'\) and \(q_2'\) is irrelevant. 
We write $C\xrightarrow{w}C'$ for a sequence $w=\delta_1\ldots\delta_k$ of transitions if there exists a sequence $C_0,\ldots,C_k$ of configurations satisfying $C=C_0\xrightarrow{\delta_1}C_1\cdots\xrightarrow{\delta_k}C_k=C'$.
We also write $C\rightarrow C'$ if $C\xrightarrow{\delta}C'$ for some transition $\delta \in \Delta$, and call $C\rightarrow C'$ an \emph{interaction}. 
We say that $C'$ is \emph{reachable from $C$} if $C\xrightarrow{w}C'$ for some (possibly empty) sequence $w$ of transitions.

Note that transitions are enabled only when there are at least two agents.
This is why we assume that populations have at least two elements.

An \emph{execution} of $\mathcal{A}$ is a finite or infinite sequence of configurations $C_0, C_1, \ldots$ such that $C_i \rightarrow C_{i+1}$ for each $i\geq 0$.
An execution $C_0, C_1, \ldots$ is {\em fair} if it is finite and cannot be extended, or it is infinite and for every step $C \rightarrow C'$, if $C_i =C$ for infinitely many indices $i \geq 0$, then $C_j = C$ and $C_{j+1} = C'$ for infinitely many indices $j \geq 0$ \cite{DBLP:conf/podc/AngluinADFP04,DBLP:journals/dc/AngluinADFP06}.
Informally, if \(C\) appears infinitely often in a fair execution, then every step enabled at \(C\) is taken infinitely often in the execution.

Given a set $S$ of configurations and a transition $t$ of a protocol scheme $(Q, \Delta)$, we define:
\begin{itemize}
\item $\post[t](S)  \defeq \{ C' \mid C \trans{t} C' \mbox{ for some $C \in S$} \}$ and $\post(S) \defeq \bigcup_{t \in \Delta} \post[t](S)$.
\item $\post^0(S) \defeq S$; $\post^{i+1}(S) \defeq \post(\post^i(S))$ for every $i \geq 0$; and  $\post^*(S) \defeq \bigcup_{i \geq 0} \post^i(S)$.
\end{itemize}
We also define $\pre[t](S)  \defeq \{ C' \mid C' \trans
{t} C \mbox{ for some $C \in S$} \}$. The sets
$\pre(S)$ and $\pre^*(S)$ are defined as above for $\post$.

\subsubsection{Immediate Observation Protocol Schemes }

A protocol scheme is \emph{immediate observation} (IO) if all its transitions are immediate observation.
A transition \((q_1,q_2) \mapsto (q'_1, q'_2)\) is immediate observation if{}f \(\{q_1,q_2\} \cap \{q'_1,q'_2\} \neq \emptyset\).
Consider, for instance, a transition \( (q_s,q_o,q_d,q_o) \) where \(q_s, q_o\) and \(q_d\) are all distinct. 
Observe that the transition is immediate observation since \(\{q_s,q_o\} \cap \{q_d,q_o\} = \{q_o\} \neq \emptyset\).
Intuitively, in an interaction specified by an immediate observation transition, one agent observes the state of another and updates it own state, but the observed agent remains as it was (and its state, unmodified by the interaction, is given by \(\{q_1,q_2\} \cap \{ q'_1,q'_2 \}\)). 
Other typical examples of immediate observation transitions are \((q_o,q_o,q_d,q_o)\), \((q_s,q_o,q_o,q_o)\) \( (q_s,q_o,q_s,q_o)\) and \( (q_o,q_o,q_o,q_o)\) where \(q_s, q_o\) and \(q_d\) are all distinct.
Note that in the last two cases, the state of two agents are the same before and after interacting.

\subsection{Population Protocols}

As Angluin \etal\cite{DBLP:conf/podc/AngluinADFP04}, we consider population protocols as a computational model, computing predicates $\Pi\colon \pop{\Sigma} \rightarrow \set{0,1}$, where $\Sigma$ is a non-empty, finite set of \emph{input variables}.

An \emph{input mapping} for a protocol scheme $\mathcal{A}$ is a function $I \colon \pop{\Sigma} \rightarrow \pop{Q}$ that maps each input population $X \in \pop{\Sigma}$ to a configuration of $\mathcal{A}$.
The set of \emph{initial configurations} is $\mathcal{I} = \{I(X) \mid X\in \pop{\Sigma}\}$. 
An input mapping $I$ is \emph{Presburger} if the set of pairs $(X, C) \in \pop{\Sigma} \times \pop{Q}$ 
such that $C=I(X)$ is definable in Presburger arithmetic.
An input mapping \(I\) is \emph{simple} if there is an injective map $\nu\colon \Sigma \rightarrow Q$ such that $I(X) = \sum_{\sigma\in \Sigma} X(\sigma) \vec{\nu(\sigma)}$.
That is, each input variable is assigned a (distinct) state, and a population $X$ over $\Sigma$ is assigned the initial configuration consisting of $X(\sigma)$ agents in the state $\nu(\sigma)$ and no other agents.
Unless otherwise specified, we restrict our attention to the class of \emph{simple} input mappings.

An \emph{output mapping} for a protocol scheme is a function $O\colon Q\rightarrow\{0,1\}$ that associates to each state $q$ of $\mathcal{A}$ an output value in $\{0,1\}$.
The output mapping induces the following properties on configurations: a configuration $C$ is a
\begin{itemize}
	\item \emph{$b$-consensus} for $b \in \{0,1\}$ if \(\sum_{p\in O^{-1}(1-b) } C(p) = 0\) and a \emph{consensus} if it is a \(b\)-consensus for some \(b\);
	\item \emph{dissensus} if it is a \(b\)-consensus for no \(b\) (that is \(C\) is a dissensus if \(\sum_{p\in O^{-1}(b) } C(p) > 0\) and \(\sum_{p\in O^{-1}(1-b) } C(p) > 0\)).%
\end{itemize}

A {\em population protocol} is a triple $(\mathcal{A}, I,O)$, where $\mathcal{A}$ is a protocol scheme, $I$ is a simple input mapping, and $O$ is an output mapping.
The population protocol is \emph{immediate observation} (IO) if $\mathcal{A}$ is immediate observation.

An execution $C_0,C_1,\ldots$ \emph{stabilizes to $b$} for a given $b\in\{0,1\}$ if 
there exists $n\in\N$ such that $C_m$ is a \(b\)-consensus for every $m\geq n$ (if the execution is finite, then this means for every $m$ between $n$ and the length of the execution). 
Notice that there may be many different executions from a given configuration $C_0$, each of which may stabilize to $0$ or to $1$ or not stabilize at all (by visiting infinitely many dissensus or infinitely many \(0\) and \(1\) consensus). 

A population protocol $(\mathcal{A},I,O)$ is {\em well-specified} if for every input configuration $C_0\in \mathcal{I}$, every fair execution of $\mathcal{A}$ starting at $C_0$ stabilizes to the same value \(b \in \{0,1\}\). 
Otherwise, it is \emph{ill-specified}.  
The {\em well specification problem} asks if a given population protocol is well-specified?

Finally, a population protocol $(\mathcal{A},I,O)$ \emph{computes} a predicate $\Pi \colon \pop{\Sigma} \rightarrow \set{0,1}$ if for every $X\in\pop{\Sigma}$, every fair execution of $\mathcal{A}$ starting at $I(X)$ stabilizes to $\Pi(X)$. 
It follows easily from the definitions that a protocol computes a predicate if{}f it is well-specified. 
The \emph{correctness} problem asks, given a population protocol and a predicate whether the protocol computes the predicate.

\section{Counting Constraints and Counting Sets}

\begin{definition}
Let $X=\{x_1, \ldots, x_n\}$ be a set of variables, and let $x \in X$.
A constraint of the form
$l \leq x$, where $l \in \N$, is a \emph{lower bound}, and a constraint of the form
$x \leq u$, where $u \in \N \cup \{ \infty\}$, is an \emph{upper bound}.
A \emph{literal} is a lower bound or an upper bound. 

A \emph{counting constraint} is a boolean combination of literals. 
A counting constraint is in \emph{counting normal form} (CoNF) if it 
is a disjunction of conjunctions of literals, where each conjunction, called a \emph{counting minterm},  
contains exactly two literals for each variable, one of them an upper bound and the other a lower bound. We often write a counting constraint in CoNF as the set of its counting minterms.
\end{definition}
The semantics of a counting constraint is a \emph{counting set}, 
a set of vectors in $\N^n$ or, equivalently,  a set of valuations to the variables in $X$.
The semantics is defined inductively on the structure of a counting constraint, as expected. 
Define $\sem{l\leq x} = \set{x\mapsto m\in\N \mid m\geq l}$ (\(\sem{\infty \leq x}=\emptyset\)) and $\sem{x \leq u} = \set{x\mapsto m\in\N \mid m\leq u}$.
Disjunction, conjunction, and negation of counting constraints translates into union, intersection, and complement of counting sets.  

The following proposition follows easily from the definition of counting sets and the disjunctive normal form for propositional logic.

\begin{proposition}
\label{prop:conf}
\hspace{0pt}
\begin{enumerate}
	\item Counting sets are closed under Boolean operations.
	\item Every counting constraint is equivalent to a counting constraint in CoNF.
\end{enumerate}
\end{proposition}
\begin{proof}[Proof Sketch.] 
	\textsf{\textbf{1.}} Proof is easy. \textsf{\textbf{2.}} Put the constraint in disjunctive normal form. Remove negations in front of literals 
using
$\sem{\neg (x_i \leq c)}=\sem{x_i \geq c+1}$ if \(c\in\N\) and remove the enclosing minterm otherwise; and $\sem{\neg (x_i \geq c)}=\sem{x_i \leq c-1}$ if \(c\in\N\setminus \{ 0\}\) and remove the enclosing minterm otherwise.
Remove minterms containing unsatisfiable literals $l \leq x_i \wedge x_i \leq u$ with $l > u$. 
Remove redundant bounds, e.g., replace $(l_1 \leq x \wedge l_2 \leq x)$ by $\max\{l_1, l_2\} \leq x$. If a minterm does not contain a lower bound (upper bound) for $x_i$, add $0 \leq x_i$
($x_i \leq \infty$).
\end{proof}

Next, we introduce a representation of CoNF-constraints used in the rest of the paper.
\begin{definition}[Representation of CoNF-constraints]
\label{def:rep}
We represent a counting minterm by a pair $M \defeq (L, U)$  where $L \colon X \rightarrow \N$ and $U \colon X \rightarrow \N \cup \{\infty\}$ assign to each variable its lower and upper bound, respectively. 
We represent a CoNF-constraint \(\Gamma\) as the set of representations of its minterms: \(\Gamma= \{ M_1,\ldots, M_m \}\).
\end{definition}

\begin{definition}[Measures of counting constraints]
	The \emph{$L$-norm} of a counting minterm $M=(L,U)$ is $\lnorm{M} \defeq \sum_{x \in X} L(x)$, and its \emph{$U$-norm} is $\unorm{M} \defeq \sum_{\substack{x\in X\\ U(x)<\infty}} U(x)$ (and \(0\) if \(U(x)<\infty\) for no \(x\)). The $L$- and $U$-norms of a CoNF-constraint $\Gamma = \{M_1, \ldots, M_m\}$  
are $\lnorm{\Gamma} \defeq \max_{i\in [1,m]} \{ \lnorm{M_i} \}$  and $\unorm{\Gamma} \defeq \max_{i\in[1,m]} \{ \unorm{M_i} \}$. 
\end{definition}

\begin{proposition}\label{prop:oponconf}
Let $\Gamma_1, \Gamma_2$ be CoNF-constraints over $n$ variables.
\begin{itemize}
\item There exists a CoNF-constraint $\Gamma$ with $\sem{\Gamma} = \sem{\Gamma_1} \cup \sem{\Gamma_2}$ such that
$\unorm{\Gamma} \leq \max \{\unorm{\Gamma_1}, \unorm{\Gamma_2} \}$ and $\lnorm{\Gamma} \leq \max \{\lnorm{\Gamma_1}, \lnorm{\Gamma_2} \}$.
\item  There exists a CoNF-constraint $\Gamma$ with $\sem{\Gamma} = \sem{\Gamma_1} \cap \sem{\Gamma_2}$ such that
$\unorm{\Gamma} \leq \unorm{\Gamma_1} + \unorm{\Gamma_2}$ and $\lnorm{\Gamma} \leq \lnorm{\Gamma_1} + \lnorm{\Gamma_2}$.
\item There exists a CoNF-constraint $\Gamma$ with $\sem{\Gamma} = \N^n \setminus \sem{\Gamma_1}$ such that
$\unorm{\Gamma} \leq n\lnorm{\Gamma_1}$ and $\lnorm{\Gamma} \leq n\unorm{\Gamma_1} + n$.
\end{itemize}
 \end{proposition}

\begin{proof}
Remember that a CoNF constraint for $m$ minterms in dimension $n$ is a $m$-disjunction of $n$-conjunctions, and that the $L$-norm (respectively $U$-norm) is the maximum sum of lower (resp. upper) bounds in one conjunction. The union of two counting sets $\Gamma_1,\Gamma_2$ with CoNF constraints is represented by the disjunction of the two constraints, and it is still CoNF so the result follows. The intersection is represented by a conjunction of the two constraints and so is not CoNF and needs to be rearranged as in Proposition \ref{prop:conf}. The new $n$-conjunctions of literals (i.e. the new minterms) mix unmodified bounds from $\Gamma_1$ and $\Gamma_2$, so the result follows. The complement is represented by the negation of the original constraint, which we rearrange into CoNF using $\lnot (l \leq x \leq u) \equiv (0\leq x \leq l -1) \vee (u+1\leq x \leq \infty)$. We obtain $n$-conjunctions with lower bounds of the form $u+1$, with $u \leq \unorm{\Gamma_1}$ an upper bound in a minterm of the original constraint. This yields $\lnorm{\Gamma} \leq n\unorm{\Gamma_1} + n$ and the reasoning is similar for the $U$-norm.
\end{proof}

\begin{remarki}
\label{rem:specialcases}
The counting sets contain the finite, upward-closed and downward-closed sets:
\begin{itemize}
\item Every finite subset of $\N^n$ is a counting set.
Indeed, $\{ (k_1, \ldots, k_n) \} = \sem{(L,U)}$ with $L(x_i)=k_i=U(x_i)$ for every $x_i \in X$,
and so finite sets are counting sets too. 
\item A set $S \subseteq \N^n$ is upward-closed if whenever $v\in S$ and $v\leq_{\times} v'$, we have $v'\in S$, where we write $v \leq_{\times} v'$ if the ordering holds pointwise (meaning $v(x) \leq v'(x)$ for every $x \in X$).
Upward-closed sets are counting sets.
Indeed, by Dickson's lemma, every upward-closed set has a finite set $\{v_1, \ldots, v_k\}$ of minimal elements with respect to $\leq_{\times}$, and so the set is $\sem{\{(L_1, U), \ldots, (L_k, U) \}}$ where $L_i(x_j) = v_i(j)$ and $U(x_j) = \infty$ for every $1 \leq j \leq n$. 
\item  A set $S\subseteq \N^n$ is downward-closed if whenever $v\in S$ and $v'\leq_{\times} v$, we have $v'\in S$. Since a set is downward-closed if{}f its complement is upward-closed, every downward-closed set is a counting set. Further, it is easy to see that downward-closed sets are represented by counting constraints $\{(L, U_1), \ldots, (L, U_k) \}$ where $L(x_j) = 0$ for every \(1 \leq j \leq n\).
\end{itemize}
\end{remarki}

Next, we define a well-quasi-ordering on counting sets.
For two counting minterms $M_1$ and $M_2$, we write $M_1 \preceq M_2$ if $\sem{M_1} \supseteq \sem{M_2}$.
For CoNF-constraints $\Gamma_1$ and $\Gamma_2$, define the ordering $\Gamma_1\sqsubseteq \Gamma_2$ if for each counting minterm
$M_2\in \Gamma_2$ there is a counting minterm $M_1\in \Gamma_1$ such that $M_1\preceq M_2$.
Note that $\Gamma_1\sqsubseteq \Gamma_2$ implies $\sem{\Gamma_1} \supseteq\sem{\Gamma_2}$.

\begin{theorem}
\label{thm:wqo}
For every $u \geq 0$, the ordering \(\sqsubseteq\) on counting sets represented by CoNF-constraints of $U$-norm at most $u$ is a well-quasi-order.
\end{theorem}
\begin{proof}
We first prove that counting minterms with $\preceq$ form a better quasi order.
For two counting minterms $M_1$ and $M_2$, we write $M_1 \preceq M_2$ if $\sem{M_1} \supseteq \sem{M_2}$.
Let ${\cal M} = M_1, M_2, \ldots$ be an infinite sequence of counting minterms of $U$-norm at most $u$, where $M_i=(L_i, U_i)$. 
Since there are only finitely many mappings $U \colon X \rightarrow \N \cup \{\infty\}$ of norm at most $u$, 
the sequence ${\cal M}$ contains an infinite subsequence ${\cal M}'$ such that every minterm $M_i$ of ${\cal M}'$ satisfies $U_i=U$ for some mapping $U$. 
So ${\cal M}'$ is of the form $(L_1,U),  (L_2, U) \ldots$
By Dickson's lemma, there are $i < j$ such that $L_i \leq_{\times} L_j$ , and so $\sem{(L_i, U)}  \supseteq \sem{(L_j,U)}$.
Hence, defining $C$ be the set of all counting minterms of $U$-norm at most $u$ we find that $(C, \preceq)$ is a well-quasi-order.
In fact, standard arguments show that this is a better-quasi-order \cite{AbdullaNylenLICS00}.
Hence, the ordering $\sqsubseteq$ is a better quasi order on counting constraints \cite{AbdullaNylenLICS00}, implying it is also a well-quasi-order.
\end{proof}
 
\section{Reachability Sets of IO Population Protocols}

We show that if $S$ is a counting set, then $\post^*(S)$ and $\pre^*(S)$ are also counting sets. 
First we show that we can restrict ourselves to IO protocols in a certain normal form.

\subsection{A Normal Form for Immediate Observation Protocols}%
\label{sub:a_normal_form_for_immediate_observation_protocols}

An IO protocol is in \emph{normal form}  if \(q_s \neq q_o\) for every transition \((q_s, q_o) \mapsto (q_o, q_d)\), i.e., the state of the 
observed agent is different from the source state of the observer. 

Given an  IO population protocol \(\PP=(\mathcal{A},I,O)\) we define an IO protocol in normal form  \(\PP'=(\mathcal{A}', I', O')\) 
which is well-specified if{}f \(\PP\) is well-specified.  Further, the number of states and transitions of $\PP'$ is linear in the number of states 
and transitions of $\PP$. The mapping \(I'\) is a Presburger mapping even if \(I\) is simple, but this does not affect our results.

\(\PP'\) is defined adding transition and states to \(\PP\).  First we add a state \(r\).
Then, we replace each transition \(t=(q,q) \mapsto (q,q_d)\) of \(\PP\) by a transition 
\(t'=(q',q) \mapsto (q',q_d)\), where \(q'\) is a primed copy of \(q\), and add 
two further transitions \( (q,r) \mapsto (r,q') \) and \( (q',r) \mapsto (r,q) \).

It remains to define the output function of the new states as well as the input mapping \(I'\) of \(\PP'\).
We define \(I'\) to be a Presburger initial mapping which coincides with \(I\) on the state of $\PP$ and 
such that \(I(X)(r)=1\) for all \(X\) and \(I(X)(q')=0\) for all \(X\) and primed state \(q'\).
The output of primed copies is the same as their unprimed version, that is  \(O(q')=O(q)\).
The only technical difficulty is the definition of the output of state $r$. Because of the way in which we have defined
the transitions involving $r$, the agent initially in state $r$ cannot leave $r$. Therefore, whatever the output $O(r)$ we assign to $r$,
the protocol $\PP'$ can never reach consensus $1- O(r)$, and so \(\PP'\) may not be well-specified even if $\PP$ is. 
To solve this problem, we add a primed copy \(r'\) of \(r\) such that \(r\) and \(r'\) have distinct outputs.
Every transition with \(r\) as observer is duplicated but this time with \(r'\) as observed state.  
Finally, for every state \(q\) of \(\PP\), if \(O(q)=O(r')\)  we add the transition
\( (q,r) \mapsto (q,r')\), and  otherwise we add the transition \( (q,r') \mapsto (q,r)\).
After adding these states, the agent initially in $r$ switches between $r$ and $r'$, and finally stabilizes to the same value the other agents stabilize to.

\subsection{The Functions \texorpdfstring{$\pre^*$}{pre*} and \texorpdfstring{$\post^*$}{post*} Preserve Counting Sets}

We show that if $S$ is a counting set, then $\post^*(S)$ and $\pre^*(S)$ are also counting sets. 
Further, given a CoNF-constraint $\Gamma$ representing $S$, we show how to construct a CoNF-constraint
representing $\post^*(S)$ and $\pre^*(S)$. 
In the following, we abbreviate $\post(\sem{\Gamma})$ to $\post(\Gamma)$, and similarly for other notations
involving $\post$ and $\pre$, like $\post[t](\Gamma)$, $\post^*(\Gamma)$, etc. 

We start with some simple examples. 
First, we observe that the result does not hold for arbitrary population protocols. 
Consider the protocol with four distinct states $\{q_1, q_2, q_3, q_4\}$ and one single transition $(q_1, q_2) \mapsto (q_3, q_4)$. Let $M = \sem{0\leq x_3 \leq 0 \wedge 0\leq x_4 \leq 0}$. Then $\post^*(M) = \sem{x_3 = x_4}$, which is not 
a counting set. Intuitively, the reason is that the transitions links the number of agents in states $x_3$ and $x_4$. However, this is only possible because the transition is not IO. Indeed, consider now the protocol $\PP_1$ with states $\{q_1, q_2, q_3\}$ and one single IO transition $(q_1, q_2) \mapsto (q_1, q_3)$. Table \ref{tab:example} lists some typical constraints for $M$, and gives constraints for 
$\post^*(M)$.

\begin{table}[t]
\newcommand{\T}{\rule{0pt}{2.6ex}} %
\newcommand{\B}{\rule[2.6ex]{-2.6pt}{2.6pt}} %
\centering
{\small
	\begin{tabular}{>{$}c<{$}|@{}>{$}c<{$}@{}|@{}>{$}c<{$}@{}|>{$}c<{$}|@{}>{$}c<{$}@{}|@{}>{$}c<{$}@{}}
     \hline 
		 M & \lnorm{M} & \unorm{M} & \Gamma\defeq\post^*[t](M) \text{ where } t\defeq (q_1, q_2) \mapsto (q_1, q_3) & \lnorm{\Gamma} & \unorm{\Gamma}\B \\ 
	\hline \T \B x_1 = 0 \wedge x_2 \geq 2 \wedge x_3 =1 & 3 & 1 & x_1 = 0 \wedge x_2 \geq 2 \wedge x_3 =1 &  3 & 1\\[0.1cm]
	x_1 = 1 \wedge x_2 = 2 \wedge x_3 \geq 1 & 4 & 3  & \begin{array}[t]{c} (x_1 =1 \wedge x_2 = 2 \wedge x_3 \geq 1) \\ \vee (x_1 =1 \wedge x_2 = 1 \wedge x_3 \geq 2) \\ \vee  (x_1 =1 \wedge x_2 = 0 \wedge x_3 \geq 3) \end{array} &  4 & 3 \\[0.1cm]
	x_1 = 1 \wedge x_2 \geq 1 \wedge x_3 =2 & 4 & 3 &  \begin{array}[t]{c} (x_1 = 1 \wedge x_2 \geq 1 \wedge x_3 = 2) \\ \vee  (x_1 = 1 \wedge x_2 \geq 0 \wedge x_3 \geq 3) \end{array} & 4 & 3 \\[0.1cm]
	x_1 \geq 0  \wedge x_2 \geq 1 \wedge x_3 \geq 2 & 3 & 0 & \begin{array}[t]{c}
		(x_1 \geq 0  \wedge x_2 \geq 1 \wedge x_3 \geq 2) \\ \vee (x_1 \geq 1  \wedge x_2 \geq 0 \wedge x_3 \geq 3)
	\end{array}  & 4 & 0 \\[0.1cm]  
	\hline\T
	M & \lnorm{M} & \unorm{M} & \Gamma\defeq\post^*[t](M) \text{ where } t\defeq (q_1, q_2) \mapsto (q_2, q_2) & \lnorm{\Gamma} & \unorm{\Gamma} \B\\ 
     \hline \T 
		 x_1 \geq 1 \wedge x_2 = 0 & 1  & 0 & x_1 \geq 1 \wedge x_2 = 0 & 1 & 0 \\[0.1cm]
		 x_1 = 1 \wedge x_2 \geq 2  & 3 & 1 & (x_1 = 1 \wedge x_2 \geq 2) \vee (x_1=0 \wedge x_2 \geq 3) & 3 & 1 \\[0.1cm]
		 x_1 \geq 2 \wedge x_2 = 1  & 3 & 1 & \begin{array}[t]{c} (x_1 \geq 2 \wedge x_2 \geq 1) \vee  (x_1 \geq 1 \wedge x_2 \geq 2) \\ \vee (x_1 \geq 0 \wedge x_2 \geq 3) \end{array} & 3 & 0 \\ 
	\hline 
	\end{tabular}
	}
	\caption{The set \(\post^*[t](M)\) for two IO transitions and counting minterm \(M\). For conciseness and clarity we use equality constraints instead of two inequalities. }
\label{tab:example}
\end{table}

Given a minterm $(L, U)$, we syntactically define a CoNF-constraint  $\fire{(L, U)}$ for the set:
\[\post^*[t](L,U) \defeq \{ C' \mid \exists k \geq 0 \exists C \in \sem{(L, U)} \text{ such that } C \trans{t^k} C' \}\enspace .\]
\noindent %
That is, $\fire{(L, U)}$ captures the set of all configurations that can be obtained from $(L,U)$ by firing transition $t$ an arbitrary number of times.

\begin{definition}
\label{def:fire}
Let $(L, U)$ be a minterm and let $t = (q_s, q_o) \mapsto (q_d, q_o)$ be an IO transition. 
Define $\fire{(L, U)}$ to be the set given by \( (L,U) \) and all the minterms \((L',U')\) such that all the following conditions hold:
\begin{enumerate}
\item \(\sem{(L'',U)}\neq\emptyset\) where \(\sem{L''} = \sem{L} \cap \sem{x_s\geq 1\land x_o\geq 1} \).
\item $U'(x) = U(x)$ and $L'(x) = L''(x)$ for every $x \in X \setminus \{x_s, x_d\}$.
\item If $U(x_s) < \infty$, then there exists $1\leq k \leq U(x_s)$ such that $U'(x_s) = U(x_s)-k$, $L'(x_s) = \max\{ 0, L''(x_s) - k\}$,
$U'(x_d) = U(x_d)+k$ and $L'(x_d) = L''(x_d)+k$. 
\item If $U(x_s) = \infty$, then $U'(x_s) = U'(x_d) =\infty$ and there exists $1\leq k \leq L''(x_s)$ such that $L'(x_s) = L''(x_s) - k$
and $L'(x_d) = L''(x_d)+k$.
\end{enumerate}
Given a CoNF-constraint $\Gamma=\{M_1, \ldots, M_m\}$, we define $\fire{\Gamma} = \bigcup_{i=1}^m \fire{M_i}$.
\end{definition}

\begin{lemma}
\label{lem:postsize}
Let $\PP$ be an IO protocol and let $\Gamma$ be a CoNF-constraint. Then $\fire{\Gamma} = \post^*[t](\Gamma)$.
Further, $\unorm{\fire{\Gamma}} \leq \unorm{\Gamma}$.
\end{lemma}
\begin{proof}
It suffices to prove that for every minterm $(L, U)$ and for every transition $t$ we have $\post^*[t](L,U)~=~\fire{(L, U)}$
and $\unorm{\fire{(L,U)}} \leq \unorm{(L,U)}$. The rest follows easily from the definitions of $\post^*$ and of a counting constraint.
 
Condition (1) holds if{}f some vector in \( \sem{(L,U)} \) enables $t$, hence \(\sem{(L'',U)}\) is the set \( \sem{(L,U)} \) of vectors minus those disabling \(t\). If no vector enables \(t\) then \(\fire{(L, U)}\) is the singleton \( \{ (L,U) \}\).
Condition (2) states that the number of agents in states other than $q_s$ and $q_d$ does not change.
Condition (3--4) defines the result of firing \(t\) one or more times.

The inequality $\unorm{\fire{(L,U)}} \leq \unorm{(L,U)}$ follows immediately from (1--4). 
Observe that $\unorm{\fire{(L,U)}} < \unorm{(L,U)}$ may hold if $U(x_s)=\infty$ and $U(x_d)< \infty$. 
\end{proof}

To prove the main theorem of the section, we introduce the following definition.

\begin{definition}\label{def:postandco}
Given a protocol $\PP$, let $S$ be a set of configurations and let \(\Gamma\) be a CoNF-constraint. 
\begin{itemize}
\item Define: $\postco(S) \defeq \bigcup_{t \in \Delta} \post^*[t](S)$; $\postco^0(S) \defeq S$ and $\postco^{i+1}(S) \defeq \postco(\postco^{i}(S))$ for every $i \geq 0$;
$\postco^*(S) \defeq \bigcup_{i \geq 0} \postco^i(S)$. 
\item Similarly, define in the constraint domain: $\postco(\Gamma) \defeq \bigcup_{t \in \Delta} \fire{\Gamma}$; $\postco^0(\Gamma) \defeq \Gamma$ and $\postco^{i+1}(\Gamma) \defeq \postco(\postco^{i}(\Gamma))$ for every $i \geq 0$.
\end{itemize}
\end{definition}
The $a$-subscript stands for ``accelerated.'' Observe that we cannot define 
$\postco^*(\Gamma)$ directly as the infinite union $\bigcup_{i \geq 0} \postco^i(\Gamma)$ because constraints 
are only closed under finite unions.

\begin{theorem}
Let $\PP$ be an IO protocol and let $S$ be a counting set. Then both $\post^*(S)$ and $\pre^*(S)$ are counting sets.
\end{theorem}
\begin{proof}
We first prove that $\post^*(S)$ is a counting set.
It follows from Definition~\ref{def:postandco} that $\post^{i}(S) \subseteq \postco^{i}(S)$ but $\postco^{i}(S) \subseteq \post^*(S)$ for every $i \geq 0$, hence $\postco^*(S) = \post^*(S)$, and so it suffices to prove that $\postco^*(S)$ is a counting set.

Let $\Gamma$ be a CoNF-constraint such that $\sem{\Gamma}=S$. 
By Lemma \ref{lem:postsize}, $\postco^{i}(\Gamma)$ is a counting set and $\unorm{\postco^{i}(\Gamma)} \leq \unorm{\Gamma}$ for every $i \geq 0$. 
By Theorem \ref{thm:wqo}, there exist indices $i < j$ such that $\postco^{j}(\Gamma) \subseteq \postco^{i}(\Gamma)$, hence \(\postco^{j}(\Gamma) = \postco^i(\Gamma)\) since \(\Gamma'\subseteq\postco(\Gamma')\) for all \(\Gamma'\), and finally $\postco^*(\Gamma) = \bigcup_{k=1}^j \postco^{k}(\Gamma)$. 
Since counting sets are closed under finite union,  $\postco^*(S)$ is a counting set.

Finally we show that $\pre^*(S)$ is also a counting set. 
Consider the protocol $\PP_r$ obtained by ``reversing'' the transitions of $\PP$, i.e., $\PP_r$ has a transition $(q_1, q_2) \mapsto (q_3, q_4)$ 
if{}f $\PP$ has a transition $(q_3, q_4) \mapsto (q_1, q_2)$. Then $\pre^*(S)$ in $\PP$ is equal to $\post^*(S)$ in $\PP_r$.
\end{proof}

\subsection{Bounding the Size of \texorpdfstring{$\post^*(\Gamma)$}{post*(L,U)}}

Given a CoNF-constraint $\Gamma$, we obtain an upper bound on the size of a CoNF-constraint denoting $\post^*(\Gamma)$ and $\pre^*(\Gamma)$. More precisely, we obtain bounds on the $L$-norm and $U$-norm of a constraint for $\post^*(\Gamma)$ as a function of the same parameters for $\Gamma$.

We first recall a theorem of Rackoff \cite{Rackoff78} recast in the terminology of population protocols.

\begin{theorem}[\cite{Rackoff78,BG11}]
\label{thm:rackoff}
Let $\PP$ be a population protocol with set of states $Q$ and let $C$ be a configuration of $\PP$. 
For every configuration $C'$, if there exists $C''$ such that $C' \trans{*} C'' \geq_{\times} C$, then there exists 
$\sigma$ and $C'''$ such that $C' \trans{\sigma} C''' \geq_{\times} C$ and $|\sigma| \leq (3 + C(Q)) ^{(3|Q|)!+1} \in C(Q)^{2^{{\cal O}(|Q| \log |Q |)}}$. 
(Recall that $C(Q) \defeq \sum_{q\in Q} C(q)$ and $C(Q) \geq 2$.)
\end{theorem}

Observe that the bound on the length of $\sigma$ depends only on $C$ and $\PP$, but not on $C'$.
Using this theorem we can already obtain an upper bounds for $\pre^*(\Gamma)$ when $\sem{\Gamma}$ is upward-closed. 
The bound is valid for arbitrary population protocols.

Recall that if $\sem{\Gamma}$ is upward-closed we can assume $\unorm{\Gamma}=0$ (see Remark \ref{rem:specialcases}).

\begin{proposition}
\label{prop:bounds-upward-closed}
Let $\PP$ be population protocol with $n$ states. Let $S$ be an upward-closed set  of configurations and let $\Gamma$ be a CoNF-constraint
with $\unorm{\Gamma}=0$ such that $\sem{\Gamma}=S$. 
There exists a CoNF constraint $\Gamma'$ such that $\sem{\Gamma'}=\pre^*(\Gamma)$ and $\unorm{\Gamma'} = 0$, $\lnorm{\Gamma'} \in (\lnorm{\Gamma})^{2^{{\cal O}(n \log n)}}$.
\end{proposition}
\begin{proof}
It is well known that if $S$ is upward-closed, then so is $\pre^*(S)$. (This follows 
from Lemma \ref{lem:postsize}, but is also an easy consequence of the fact that $C \trans{*} C'$ implies $C+C'' \trans{*} C'+C''$ for every $C''$).
Let $K \defeq (3 + \lnorm{\Gamma}) ^{(3n)!+1}$. 
By Theorem \ref{thm:rackoff}, for every configuration $C$, if $C \in \pre^*(S)$ then $C \in \bigcup_{i=0}^K \pre^i(S)$, and so $\pre^*(S)=\bigcup_{i=0}^K\pre^i(S)=\preco^K(S)$. Let $\Gamma'= \preco^K(\Gamma)$. Then $\sem{\Gamma' } = \pre^*(S)$. 
Further, we have $\unorm{\Gamma'} = 0$ by Lemma \ref{lem:postsize} (the Lemma proves the result for $\post^*$, but exactly the same proof works for $\pre^*$ by reversal of transitions). 
To prove the bound for the $L$-norm, observe that by the definition of $\fire{(L,U)}$ we have  $\lnorm{\fire{(L, U)}} \leq \lnorm{(L,U)}+1$, as we are always in case 4. of Definition \ref{def:fire} (because $S$ is upward-closed). 
Since $\preco(\Gamma) = \bigcup_{t \in \Delta_r} \fire{\Gamma}$ and the $L$-norm of a union is the maximum of the $L$-norms, we get $\lnorm{\preco(\Gamma)} \leq \lnorm{\Gamma} + 1$. 
By induction, $\lnorm{\preco^K(\Gamma)} \leq \lnorm{\Gamma} + K$, and the result follows. 
\end{proof}

In the rest of the section we obtain a bound valid not only for upward-closed sets, but for arbitrary counting sets. The price to pay is a restriction to 
IO protocols. We start with some miscellaneous notations that will be useful.

\begin{itemize}
\item Given a mapping $f \colon X \ra \N$ and $Y \subseteq X$ we write $f(Y)$ for
$\sum_{x \in Y} f(x)$, and $f|_Y$ for the projection of $f$ onto $Y$. 
\item Given a transition sequence $\sigma$, we denote by $c(\sigma)$ the ``compression'' of $\sigma$ as the shortest regular expression $r=t_1^* \ldots t_m^*$ such that $\sigma \in L(r)$, and denote $|c(\sigma)|=m$.
	While \(\sigma\) induces a sequence of \(\pre[t]\) or \(\post[t]\), \(c(\sigma)\) induces a sequence of \(\pre^*[t]\) or \(\post^*[t]\).
\end{itemize}

For the rest of the section we fix an IO protocol $\PP$ with a set of states $Q$ and $|Q|=n$.
We say that $C$ \emph{covers} $C'$ if $C \geq_\times C'$.
We introduce a relativization. 

\begin{definition}
Let $E \subseteq Q$. A configuration $C$ \emph{$E$-covers}  $C'$, denoted $C \ecov{E} C'$, if
$C(q) = C'(q)$ for every $q \in E$ and $C(q) \geq C'(q)$ for every $q \in Q \setminus E$.
$\PP$ is \emph{$E$-increasing} if  for every transition $(q_s, q_o) \mapsto (q_d, q_o)$ either $q_s \notin E$ or $q_d \in E$.
\end{definition}
Observe that $\PP$ is vacuously $\emptyset$-increasing and $Q$-increasing.  
Intuitively, if $\PP$ is $E$-increasing then the total number of agents in the states of $E$ cannot decrease. 
Indeed, for that we would need a transition that removes agents from $E$ without replacing them, i.e., 
a transition such that $q_s \in E$ and $q_d \notin E$. So, by induction, we have:

\begin{lemma}
\label{lem:noinc}
If $\PP$ is $E$-increasing and $C' \trans{*} C$ then $C'(E) \leq C(E)$.
\end{lemma}

Now we give a result bounding the length of $E$-covering sequences for $E$-increasing protocols.

\begin{lemma}
\label{lem:rackoff}
Let $\PP = (Q,\Delta)$ be an IO protocol scheme, let $C$ be a configuration of $\PP$, and let $E \subseteq Q$ such that $\PP$ is $E$-increasing. 
For every configuration $C'$, if there exists $C''$ such that $C' \trans{*} C'' \ecov{E} C$, then 
there exists $\sigma$ and $C'''$ such that $C' \trans{\sigma} C''' \ecov{E} C$ and  $|\sigma| \in C(Q)^{2^{{\cal O}(n \log n)}}$, where the constant in the Landau symbol is independent of $\PP$ and $C$.
\end{lemma}
\begin{proof}
We use a theorem of Bozzelli and Ganty \cite{BG11} that generalizes Rackoff's theorem to Vector Addition Systems with States (VASS).
Recall that a $d$-VASS is a pair $(P, \Delta)$ where $P$ is a set of control points and $\Delta \subseteq P \times \mathbb{Z}^d \times P$ is a finite
set of transitions. The number $d$ is called the dimension. A configuration of a $d$-VASS is a pair $(p, v)$, where $p \in P$ and 
$v \in \N^d$. Intuitively, the VASS acts on $d$ counters that can only take non-negative values. Formally, we have $(p, v) \ra (p',v')$ if 
there is a transition $(p, v'', p')$ such that $v + v'' = v'$, i.e., the machine moves from $p$ to $p'$ by updating the counters with $v''$.
Given two configurations $(p, v)$ and $(p',v')$, we write $(p, v) \geq_\times (p', v')$
if $p=p'$ and $v \geq_\times v'$. It is shown \cite{BG11} in Theorem 1  that given a $d$-VASS $(P, \Delta)$ and a configuration $C$, for each
configuration $C'$, if  there exists $C''$ such that $C' \trans{*} C'' \geq_{\times} C$, then there exists 
$\sigma$ and $C'''$ such that $C' \trans{\sigma} C''' \geq_{\times} C$ and $|\sigma| \leq |P| \cdot (\|\Delta\|_1 + \|C\|_1+2)^{(3d)! +1}$, where
$\|\Delta\|_1$ and $\|C\|_1$ denote the maximal components of $\Delta$ and $C$, respectively. 

Let $n = |Q|$. 
We construct a VASS $V_{\PP,E}$ that simulates the protocol $\PP$, and then apply Bozzelli and Ganty's theorem. 
We do not give all the formal details of the construction.  Intuitively, given a configuration $C$ of $\PP$, we split it into $(C|_E, C|_{Q \setminus E})$. Since $\PP$ is $E$-increasing, every configuration
$(C'|_E, C'|_{Q \setminus E})$ from which we can reach $(C|_E, C|_{Q \setminus E})$ satisfies $C'|_E(E) \leq C|_E(E)$ (Lemma \ref{lem:noinc}), and so there are only finitely many 
(at most $(C(E)+1)^n$)
possibilities for $C'|_E$. 
The control points of the VASS $V_{\PP,E}$ correspond to these finitely many possibilities. 
Formally, the set of control points of $V_{\PP,E}$ is the set of all mappings $M \colon E \ra \N$ such that $M(E) \leq C(E)$, plus some auxiliary control points (see below). The dimension, or number of counters, is $|Q \setminus E|$. The transitions of $V_{\PP,E}$ simulate the transitions of $\PP$. 
For example, assume $t = (q_o, q_s) \mapsto (q_o, q_d)$ is a transition of $\PP$ such that $q_s, q_o \notin E$ and $q_d \in E$. 
Then for every control point $M$ of $V_{\PP,E}$ the VASS has a transition $t_1$ leading from $M$ to an auxiliary control point $\langle M, t \rangle$, and a transition $t_2$ leading from $\langle M, t \rangle$ to the control point $M'$ given by $M'(q_d)=M(q_d)+1$ and $M'(q)=M(q)$ for every other $q \in E$. 
Transition $t_1$ decrements the counter of $q_s$ and $q_o$ by $1$, leaving all other counters untouched, and transition $t_2$ increments the counters $q_o$, leaving all other counters untouched. 

It follows that there is an execution $C' \trans{*} C'' \ecov{E} C$ in $\PP$ if{}f there is an execution $(C'|_E, C'|_{Q \setminus E}) \trans{*} (C''|_E, C''|_{Q \setminus E}) \geq_{\times} (C|_E, C|_{Q \setminus E})$ in $V_{\PP,E}$ of at most twice the length.

Applying Bozzelli and Ganty's theorem, we obtain that the length of $\sigma$ is bounded by $|P| \cdot (\|\hat{\Delta}\|_1 + \|C\|_1+2)^{(3d)! +1}$,
where $|P|$, $ \hat{\Delta}$, and $d$ are now the set of control points, transitions, and dimension of $V_{\PP,E}$.  
We have $|P| \leq (C(E)+1)^n + |\Delta| (C(E)+1)^n$, $d = |Q \setminus E| \leq n$, $\|\hat{\Delta}\|_1=2$. 
Further, we have $\|C\|_1 \leq C(Q\setminus E)$, which leads to a bound of $(1+ |\Delta|)(C(E)+1)^n \cdot (C(Q \setminus E)+4)^{(3n)! +1} \in C(Q)^{2^{O(n \log n)}}$.
\end{proof}

Next we prove a double exponential bound on the length of $E$-covering sequences. 
The result is similar to Lemma~\ref{lem:rackoff} with two important changes:  the restriction to $E$-increasing protocols is dropped, and we consider the bound on the length of \(c(\sigma)\) instead of \(\sigma\).

\begin{theorem}
	\label{thm:rackoff2}
Let $\PP$ be an IO protocol with a set \(Q\) of \(n\) states, and let $C$ be a configuration of $\PP$. For every $E \subseteq Q$ and for every configuration $C_0$, if there exists $\tau$ and $C'$ such that $C_0\trans{\tau}~C'\ecov{E}C$, then  there exists $\sigma$ and $C''$ such that $C_0 \trans{\sigma} C'' \ecov{E} C$ and  $|c(\sigma)| \in C(Q)^{2^{\mathcal{O}(n^2 \log n)}}$, where the constant in the Landau symbol is independent of $\PP$, $C$, and $C_0$.
\end{theorem}
\begin{proof}
We prove by induction on $|E|$ that the result holds with $|c(\sigma)| \in C(Q)^{2^{e \mathcal{O}(n \log n)}}$, where $e \defeq \max\{ 1, |E|\} $, and then apply $e \leq n$.

\smallskip
\noindent \textbf{Base:} $|E| = 0$. Then $\PP$ is vacuously $E$-increasing, and the result follows from Lemma \ref{lem:rackoff}.

\smallskip
\noindent \textbf{Step:} $|E| > 0$. We use the following notation: Given a transition sequence $\rho$, we denote $\PP_\rho$ the restriction of 
$\PP$ to the transitions that occur in $\rho$.

If $\PP_\tau$ is $E$-increasing, then we can apply Lemma \ref{lem:rackoff}, and we are done.
Else, the definition of $E$-increasing shows there exist $C_1$ and $C_2$ and a decomposition $\tau = \tau_1 \, t \, \tau_2$  such that
\[
C_0 \trans{\tau_1} C_1 \trans{t} C_2 \trans{\tau_2} C' \ecov{E} C \enspace.
\]
The protocol $\PP_{\tau_2}$ is $E$-increasing, but $\PP_{t \tau_2}$ is not $E$-increasing (observe that possibly $\tau_2 = \epsilon$). 
By Lemma \ref{lem:rackoff} applied to $\PP_{\tau_2}$, there exists  $\sigma_2$ and $\tilde{C}''$ such that 
\[
C_0 \trans{\tau_1} C_1 \trans{t} C_2 \trans{\sigma_2} \tilde{C}'' \ecov{E} C
\quad \mbox{ and } \quad |\sigma_2| \in C(Q)^{2^{{\cal O}(n \log n)}} \enspace.
\]
Since $\sigma_2$ can remove at most $|\sigma_2|$ agents from a state, there exist $C_1', C_2', C''$ such that 
\[
C_0 \trans{\tau_1} C_1 \ecov{E} C_1' \trans{t} C_2' \trans{\sigma_2} C'' \ecov{E} C \quad \mbox{ and } C_1'(Q) \in C(Q)^{2^{{\cal O}(n \log n)}} \enspace. 
\]
\noindent Indeed, it suffices to define 
	\begin{itemize}
		\item $C_1'(q) = \min \{ C_1(q), |\sigma_2| + C(q)\}$ for every $q \in Q \setminus E$ and $C_1'(q)=C_1(q)$ for every $q \in E$, 
		\item $C_2'(q) = \min \{ C_2(q), |\sigma_2| + C(q)\}$ for every $q \in Q \setminus (E \cup \{ q_d \})$, $C_2'(q)=C_2(q)$ for every $q \in E$ and  $C_2'(q_d)=\min \{ C_2(q_d), 1 + |\sigma_2| + C(q)\}$ where $t = (q_o, q_s) \mapsto (q_o, q_d)$. 
	\end{itemize} 

    Recall that $\PP_{t \tau_2}$ is not $E$-increasing, and so  $t = (q_o, q_s) \mapsto (q_o, q_d)$ for states $q_s, q_d$ such that
    $q_s \in E$ and $q_d \notin E$.  (Intuitively, the occurrence of $t$  ``removes agents'' from $E$.) Let $E' \defeq E \setminus \{q_s\}$. Since $C_0 \trans{\tau_1} C_1 \ecov{E} C_1'$, we also have
    $C_0 \trans{\tau_1} C_1 \ecov{E'} C_1'$.
    By induction hypothesis, there exists $\sigma_1$ and $C_1''$ such that $C_0 \trans{\sigma_1} C_1'' \ecov{E'} C_1'$ and 
    \begin{align*}
    |c(\sigma_1)| & \in C_1'(Q)^{2^{e'{\cal O}(n \log n)}} \in \left(C(Q)^{2^{{\cal O}(n \log n)}}\right)^{2^{e'\mathcal{O}(n \log n)}} \in C(Q)^{  2^{{\cal O}(n \log n)} \cdot 2^{e'{\cal O}(n \log n)}  } \\
    & \in C(Q)^{ 2^{ {\cal O}(n \log n) + e'{\cal O}(n \log n)} } \in C(Q)^{ 2^{ e {\cal O}(n \log n)} } \enspace .
    \end{align*}   
   \noindent 
	(Observe that $C_1'' \ecov{E'} C_1'$ holds, but $C_1'' \ecov{E} C_1'$ may  not hold, we may have $C_1''(q_s) > C_1'(q_s)$.) 
	
	\noindent
	To sum up, we have configurations $C_1', C_1'', C_2', C''$ and transition sequences $\sigma_1, \sigma_2$ such that
	\[
	C_0 \trans{\sigma_1} C_1'' \ecov{E'} C_1' \trans{t} C_2' \trans{\sigma_2} C'' \ecov{E} C \quad \mbox{ and } \quad |c(\sigma_1 \, t \, \sigma_2)| \in C(Q)^{ 2^{ e {\cal O}(n \log n)}} \enspace.
	\]	

	\noindent \textbf{Claim:} There exist $C_2''$ and $C'''$ such that
	\[
	C_0 \trans{\sigma_1} C_1''  \trans{t^{C_1''(q_s)-C_1'(q_s)+1}} C_2'' \trans{\sigma_2} C'''\ecov{E} C\enspace .
	\]	
	\noindent \textbf{Proof of the claim:} Since $C_1'' \ecov{E'} C_1'$ and $C_1'$ enables $t$, so does $C_1''$.  Since $\PP$ is an IO protocol (a hypothesis we had not used so far), $C_1''$ enables not only $t$, but also the sequence $t^{C_1''(q_s)-C_1'(q_s)+1}$. So there indeed exists a configuration $C_2''$ such that 
	\[
	C_0 \trans{\sigma_1} C_1''  \trans{t^{C_1''(q_s)-C_1'(q_s)+1}} C_2'' \ .
	\] 	
	\noindent It remains to prove that
	$C_2'' \trans{\sigma_2} C''' \ecov{E} C$ holds for some configuration $C'''$. First we show $C_2'' \ecov{E} C_2'$, which amounts to proving
	$C_2'' \ecov{E'} C_2'$ and $C_2''(q_s) = C_2'(q_s)$. 

	The first part, i.e., $C_2'' \ecov{E'} C_2'$, follows from:
	\(C_1'' \trans{t^{C_1''(q_s)-C_1'(q_s)+1}} C_2''\), \(C_1''\ecov{E'}C_1'\), 
	\(C_1' \trans{t} C_2'\),
	$q_d \notin E$, which implies $q_d \notin E'$, and the fact that $t$ move agents from $q_s$ to $q_d$ (thus increasing their number in $q_d$). 
	The second part, $C_2''(q_s) = C_2'(q_s)$, is proved by	
	\[
	C_2''(q_s) = C_1''(q_s) - (C_1''(q_s)-C_1'(q_s)+1) = C_1'(q_s)-1 = C_2'(q_s) \enspace . 
	\]
	\noindent So indeed we have $C_2'' \ecov{E} C_2'$. Now, since $C_2'$ enables $\sigma_2$ and $C_2'' \ecov{E} C_2'$, the configuration $C_2''$ enables $\sigma_2$ too. So there exists a configuration $C'''$ such that $C_2'' \trans{\sigma_2} C'''$. Further, since 	
	\(
	\begin{array}[b]{@{}c@{}c@{}c@{}c@{}c@{}c@{}c@{}}
	C_1'' & \trans{t^{C_1''(q_s)-C_1'(q_s)+1}} & C_2'' & \trans{\sigma_2} & C'''  \\
	\ecov{E'} &  & \ecov{E}\\
	C_1' & \trans{\hspace{3em}t\hspace{3em}} & C_2' & \trans{\sigma_2} & C'' & \ecov{E} & C
	\end{array}
	\)
	holds, we have
	\(
	\begin{array}[b]{@{}c@{}c@{}c@{}c@{}c@{}c@{}c@{}}
	C_1'' & \trans{t^{C_1''(q_s)-C_1'(q_s)+1}} & C_2'' & \trans{\sigma_2} & C'''  \\
	\ecov{E'} &  & \ecov{E} &  & \ecov{E}\\
	C_1' & \trans{\hspace{3em}t\hspace{3em}} & C_2' & \trans{\sigma_2} & C'' & \ecov{E} & C
	\end{array}
	\)
	So $C''' \ecov{E} C'' \ecov{E} C$, and the claim is proved. \qed
	
	By the claim we have	
	\(
	C_0 \trans{\sigma_1 \, t^{C_1''(q_s)-C_1'(q_s)+1} \, \sigma_2} C''' \ecov{E} C \enspace .
	\)
	\noindent Let $\sigma = \sigma_1 t^{C_1''(q_s)-C_1'(q_s)+1} \sigma_2$. While $C_1''(q_s)-C_1'(q_s)$ can be arbitrarily large, we have $c(\sigma)= c(\sigma_1 \, t \, \sigma_2)$, and so we conclude
	\(
	C_0 \trans{\sigma} C''' \ecov{E} C \quad \mbox{ and } \quad |c(\sigma)|  \in C(Q)^{2^{e {\cal O}(n \log n)}}
	\).
\end{proof}

Theorem \ref{thm:rackoff2} allows to derive the promised bounds on a constraint for $\pre^*(\Gamma)$ and $\post^*(\Gamma)$.

\begin{theorem}
\label{thm:closureprepost}
Let $\PP$ be an IO population protocol with $n$ states, and let $\Gamma$ be a CoNF-constraint. There exists a CoNF-constraint
$\Gamma'$ satisfying $\sem{\Gamma'} = \pre^*(\Gamma)$,  $\unorm{\Gamma'} \leq \unorm{\Gamma}$ and $\lnorm{\Gamma'} \in \unorm{\Gamma} \left( \lnorm{\Gamma} + \unorm{\Gamma} \right)^{ 2^{ {\cal O}(n^2 \log n)} }$. 
Further, $\Gamma'$ can be constructed in 
${(2+\unorm{\Gamma})}^n \cdot \unorm{\Gamma} \left( \lnorm{\Gamma} + \unorm{\Gamma} \right)^{ 2^{ {\cal O}(n^2 \log n)} }$ 
time and space. 
Further, the same holds for $\post^*(\Gamma)$.
\end{theorem}

\begin{proof}
The bound on $\unorm{\Gamma'}$ follows from Lemma \ref{lem:postsize}. The bound on $\lnorm{\Gamma'}$ is proved in a similar way to Proposition 
\ref{prop:bounds-upward-closed}, but using Theorem \ref{thm:rackoff2} instead of Theorem \ref{thm:rackoff}. 
Let $(L,U)$ be a counting minterm in $\Gamma$. We define the set of states $E_{(L,U)} = \{ q_i \mid U(x_i) < \infty \}$ and $\mathcal{C}_{(L,U)}^{\min} = \{ C \mid \forall q_i \in E_{(L,U)}, L(x_i) \leq C(q_i) \leq U(x_i) \text{ and } \forall q_i \in E_{(L,U)}, C(q_i) = L(x_i) \}$ the configurations of $(L,U)$ minimal over $Q \backslash E_{(L,U)}$. Notice that a configuration is in $(L,U)$ if and only if it covers a configuration in $\mathcal{C}_{(L,U)}^{\min}$. By applying Theorem \ref{thm:rackoff2} to every $C \in \mathcal{C}_{(L,U)}^{\min}$ and to $E_{(L,U)}$, we get
$\pre^*(L,U) = \bigcup_{i=0}^K\preco^i(L,U)$ for $K$ the bound in Theorem \ref{thm:rackoff2} but with $\left( \sum_{q_i \in Q \backslash E} L(x_i) + \sum_{q_i \in E} U(x_i) \right)$ instead of $C(Q)$.
Now since $\Gamma$ is the union of such minterms $(L,U)$, and by definition of the $L$ and $U$-norms, $\pre^*(\Gamma) = \bigcup_{i=0}^K\preco^i(\Gamma)$ for $K \in \left( \lnorm{\Gamma} + \unorm{\Gamma} \right)^{ 2^{ {\cal O}(n^2 \log n)} }$. By Definition \ref{def:fire}, we have $\lnorm{\fire{(L, U)}} \leq \lnorm{(L,U)} + ( \unorm{(L,U)} - 1 )$. Using $\unorm{\fire{\Gamma}} \leq \unorm{\Gamma}$, we reason by induction and get $\lnorm{\preco^i(\Gamma)} \leq \lnorm{\Gamma} + i ( \unorm{\Gamma} - 1 )$ for all $i$, and the result on the $L$-norm follows.

The algorithm needs linear time and space in the number of minterms of $\Gamma'$. 
An upper bound on the number of minterms $(L, U)$ is computed as follows. Since $\lnorm{\Gamma'} \in \unorm{\Gamma} \left( \lnorm{\Gamma} + \unorm{\Gamma} \right)^{ 2^{ {\cal O}(n^2 \log n)} }$, there are at most $(1+ \lnorm{\Gamma'})^n \in \unorm{\Gamma} \left( \lnorm{\Gamma} + \unorm{\Gamma} \right)^{ 2^{ {\cal O}(n^2 \log n)} }$ possibilities for $L$, and since $\unorm{\Gamma'} \leq \unorm{\Gamma}$ at most $\left(2+\unorm{\Gamma}\right)^n$ possibilities for $U$. %
\end{proof}

The following result characterizes the size of counting constraints.

\begin{corollary}
\label{coro:closure}
Let $\PP$ be an IO protocol with $n$ states. Given $c \geq 2, d \geq 1$, let ${\cal G}(c, d)$ be the class of  CoNF-constraints $\Gamma$ such that $\lnorm{\Gamma}, \unorm{\Gamma} \leq c^{2^{d \cdot (n^2 \log n)}}$. There exists a constant $k$ that does not depend on $n$ or $\PP$ such that :
\begin{enumerate}
\item for every $\Gamma_1, \Gamma_2 \in {\cal G}(c, d)$, there exists $\Gamma \in {\cal G}(c, d)$ such that $\sem{\Gamma} = \sem{\Gamma_1} \cup \sem{\Gamma_2}$.
\item for every $\Gamma_1, \Gamma_2 \in {\cal G}(c, d)$, there exists $\Gamma \in {\cal G}(c, d+1)$ such that $\sem{\Gamma} = \sem{\Gamma_1} \cap \sem{\Gamma_2}$.
\item for every $\Gamma_1 \in {\cal G}(c, d)$, there exists $\Gamma \in {\cal G}(c, d+1)$ such that $\sem{\Gamma} = \N^n \setminus \sem{\Gamma_1}$.
\item for every $\Gamma_1 \in {\cal G}(c, d)$, there exists $\Gamma \in {\cal G}(c,d+k+2)$ such that $\sem{\Gamma} = \pre^*\left(\sem{\Gamma_1}\right)$.
\item for every $\Gamma_1 \in {\cal G}(c, d)$, there exists $\Gamma \in {\cal G}(c,d+k+2)$ such that $\sem{\Gamma} = \post^*\left(\sem{\Gamma_1}\right)$.
\end{enumerate}
\end{corollary}

The first three bounds follow from Prop~\ref{prop:oponconf}.
For the last two, the constant $k$ is the one from the Landau symbol in Theorem \ref{thm:closureprepost}.

\section{An Algorithm for Deciding Well Specification}

We show that the well-specification and correctness problems can be solved in exponential space for IO protocols, 
improving on the result for general protocols stating that they are at least as hard as 
the reachability problem for Petri nets \cite{DBLP:conf/concur/EsparzaGLM15}.
We first introduce some notions.

\begin{definition}\label{def:consensusandco}
Given a population protocol $\PP$, a configuration $C$ is a \emph{stable $b$-consensus} if $C$ is a $b$-consensus and so is $C'$ for every $C'$ reachable from $C$.
Let $\mathcal{C}_b$ and $\mathcal{ST}_b$ denote the sets of $b$-consensus and stable $b$-consensus configurations of $\PP$. Observe that $\mathcal{ST}_b = \overline{\pre^*(\overline{\mathcal{C}_b})}$.
\end{definition}

Next, we characterize the well-specified protocols starting with the following lemma.
\begin{lemma}
\label{lem:auxfair}
Let $\PP$ be a population protocol, let $C_0,C_1,C_2,\ldots$ be a fair execution of $\PP$, and 
let $S$ be a set of configurations. If $S$ is reachable from $C_i$ for infinitely many indices $i \geq 0$, then $C_j \in S$ for infinitely many indices $j \geq 0$.
\end{lemma}
\begin{proof}
Let $n$ be the number of states of $\PP$ and let $m$ be the number of agents of $C_0$. Then there are 
at most $K \defeq (m+1)^n$ configurations reachable from $C_0$. 
So for infinitely many indices $i \geq 0$ we have $C_i \in \cup_{i\leq K} \pre^i(S)$. 
We proceed by induction on $K$. If $K = 0$, then $C_i \in S$ and we are done. 
If $K > 0$, then by fairness there exist infinitely many indices $j \geq 0$ such that $C_j \in \cup_{i\leq K-1} \pre^{i}(S)$, and 
we conclude by induction hypothesis.
\end{proof}

\begin{proposition}
\label{prop:decproc}
A population protocol $\PP$ is well-specified if{}f the following hold:
\begin{enumerate}
	\item \(\post^*(\mathcal{I}) \subseteq \pre^*(\mathcal{ST}_0 \cup \mathcal{ST}_1 ) \) (or, equivalently, \(\post^*(\mathcal{I}) \cap \overline{ \pre^*(\mathcal{ST}_0)}  \cap \overline{\pre^*(\mathcal{ST}_1 )} = \emptyset  \) ); 
	\item \(\pre^*(\mathcal{ST}_0 ) \cap \pre^*(\mathcal{ST}_{1} ) \cap  \mathcal{I} = \emptyset\).
\end{enumerate}
\end{proposition}
\begin{proof}
We start with \(\mathcal{ST}_b\) which is defined (Definition \ref{def:consensusandco}) as the set of configurations $C$ such that \(C\) is a $b$-consensus and so is $C'$ for every $C'$ reachable from $C$.

By definition, \(\PP\) is \emph{well-specified} if for every input configuration $C_0\in \mathcal{I}$, every fair execution of $\PP$ starting at $C_0$ stabilizes to the same value \(b \in \{0,1\}\). Equivalently,  \(\PP\) is \emph{well-specified} if every input configuration $C_0\in \mathcal{I}$ satisfies the following two conditions:
\begin{alphaenumerate}
	\item every fair execution starting at $C_0$ stabilizes to some value; and
\item no two fair executions starting at $C_0$ stabilize to different values (i.e., to \(0\) and to \(1\) ). 
\end{alphaenumerate}
We claim that \textsf{\textbf{(a)}} is equivalent to:\\
\hspace*{\stretch{1}}for every $C \in \post^*(\mathcal{I})$ there exists $C'$ such that $C \trans{*} C'$ and  $C' \in \mathcal{ST}_0 \cup \mathcal{ST}_1$.\hfill \textsf{\textbf{(A)}}  

Assume \textsf{\textbf{(a)}} holds, and let $C \in \post^*(\mathcal{I})$. Then $C_0 \trans{*} C$ for some $C_0 \in \mathcal{I}$. Extend $C_0 \trans{*} C$ to a fair execution. By \textsf{\textbf{(a)}}, the execution stabilizes to some value $b$. So $\mathcal{ST}_b$ is reachable from every configuration of the execution.  By Lemma \ref{lem:auxfair}, the execution reaches a configuration $C' \in \mathcal{ST}_b$. For the other direction, assume \textsf{\textbf{(A)}} holds, and consider a fair execution starting at $C_0 \in \mathcal{I}$. By Lemma \ref{lem:auxfair}, the execution reaches a configuration of $\mathcal{ST}_b$ for $b \in \{0,1\}$. By the definition of  $\mathcal{ST}_b$, all successor configurations also belong to $\mathcal{ST}_b$, and so the execution stabilizes to $b$. 
Now we claim that \textsf{\textbf{(b)}} is equivalent to:\\
\hspace*{\stretch{1}}no configuration $C \in \post^*(\mathcal{I})$ can reach both $\mathcal{ST}_0$ and $ \mathcal{ST}_1$.\hfill \textsf{\textbf{(B)}}  

Assume \textsf{\textbf{(B)}} does not hold, i.e.,  there is $C \in \post^*(\mathcal{I})$
and configurations $C_0 \in \mathcal{ST}_0$ and $C_1 \in \mathcal{ST}_1$ such that $C \trans{*} C_0$ and $C \trans{*} C_1$. These two executions can be extended to fair executions, and by the definition of $\mathcal{ST}_0$ and $\mathcal{ST}_1$ these executions stabilize to $0$ and $1$, respectively. 
So \textsf{\textbf{(b)}} does not hold.

Assume now that \textsf{\textbf{(b)}} does not hold. Then two fair executions starting at $C_0$ stabilize to different values.
So $C_0$ can reach both  $\mathcal{ST}_0$ and $ \mathcal{ST}_1$, and \textsf{\textbf{(B)}}   does not hold.

So \textsf{\textbf{(a)}} and \textsf{\textbf{(b)}} are equivalent to \textsf{\textbf{(A)}}   and \textsf{\textbf{(B)}}. Since \textsf{\textbf{(A)}}   is equivalent to \(\post^*(\mathcal{I}) \subseteq \pre^*(\mathcal{ST}_0 \cup \mathcal{ST}_1 ) \), and \textsf{\textbf{(B)}}  is equivalent to \(\pre^*(\mathcal{ST}_0 ) \cap \pre^*(\mathcal{ST}_{1} ) \cap  \mathcal{I} = \emptyset \), we are done.
\end{proof}

\begin{theorem}
\label{thm:main}
The well specification problem for IO protocols is in \EXPSPACE\ and is \PSPACE-hard.
\end{theorem}
\begin{proof}

Let $\PP$ be an IO protocol with \(n\) states. 
Recall that \(\mathcal{ST}_b\) is given by \(\overline{ \pre^*( \overline{ \mathcal{C}_b } ) }\) where $\mathcal{C}_b$, for $b \in \{ 0,1 \}$, can be represented by the CoNF-constraint of single minterm defined by $x_i=0$ for all $q_i\in O^{-1}(1-b)$ and $0 \leq x_i \leq \infty$ otherwise. 
By Corollary~\ref{coro:closure}, there exists a constant $d$, independent of $\PP$, and a CoNF constraint $\Gamma \in {\cal G}(2, d)$ such that \(\sem{\Gamma}\) is given by \(\post^*(\mathcal{I}) \cap \overline{ \pre^*(\mathcal{ST}_0)}  \cap \overline{\pre^*(\mathcal{ST}_1 )}\).

In order to falsify condition \textsf{\textbf{1.}}   of Proposition~\ref{prop:decproc} it suffices to exhibit, following the previous reasoning, a “small” configuration \(C\), such that \(C(Q) \leq c^{2^{d\cdot (n^{2} \log n)}} \), in the intersection. 
Note that \(C\) can be written in \EXPSPACE.  
The \EXPSPACE{} decision procedure follows the following steps:
\textsf{\textbf{1.}}    Guess a “small” configuration \(C\).
\textsf{\textbf{2.}}  Check that \(C\) belongs to \(\post^*(\mathcal{I})\).
\textsf{\textbf{3.}}  Check that \(C\) belongs to \(\overline{\pre^*(\mathcal{ST}_b)}\), for \(b=0,1\).

\noindent
Algorithm for \textsf{\textbf{2.}}: Guess a at most double exponential sequence of minterms such that the first one  is a minterm of \(\mathcal{I}\), and every pair of consecutive minterms is related by \(\post^*[t]\) (given by Definition~\ref{def:fire}) for some \(t\). Observe that we keep track of the last computed element and the number of steps performed so far in exponential space. Then, check that \(C\) belongs to the resulting minterm. 

\noindent
Algorithm for \textsf{\textbf{3.}}: 
It follows from \(\EXPSPACE = \textsf{coEXPSPACE}\) that it is equivalent to check \(C \in \pre^*(\mathcal{ST}_b)\) is in \EXPSPACE.
Our algorithm is divided in two steps.

\noindent 
Step 1. Let \(c,d\) be such that \(\mathcal{ST}_b \in \mathcal{G}(c,d)\).
Guess a minterm \(M\) in \(\mathcal{G}(c,d)\) and proceed similarly to Algorithm for \textsf{\textbf{2.}} to compute a minterm of \(\pre^*(M)\) and then check that \(C\) belongs to the resulting minterm.

\noindent
Step 2. Verify that \(M\) does indeed belong to \(\mathcal{ST}_b\).
Formally, we rely on the following equivalences:
\(\sem{M} \subseteq \mathcal{ST}_b \) if{}f \(\sem{M} \subseteq \overline{ \pre^*(\overline{\mathcal{C}_b})}\) if{}f \(\sem{M} \cap  \pre^*(\overline{\mathcal{C}_b}) = \emptyset\).
Using \(\EXPSPACE = \textsf{coEXPSPACE}\) we now show that \(\sem{M} \cap  \pre^*(\overline{\mathcal{C}_b}) \neq \emptyset\) belongs to \(\EXPSPACE\).
We nondeterministically choose a minterm in \(\overline{\mathcal{C}_b}\) and as previously explained guess a minterm in \(\pre^*(\overline{\mathcal{C}_b})\).
Finally, we check whether it intersects with \(\sem{M}\).

We use a similar reasoning for checking in \EXPSPACE{} condition \textsf{\textbf{2.}} of Proposition~\ref{prop:decproc}. 

The proof for \PSPACE-hardness reduces from the acceptance problem for 
deterministic Turing machines running in linear space \cite{Papadimitriou}.
The proof follows the structure of analogous proofs for 1-safe Petri nets \cite{JonesLL77} (and also \cite{ChengEP95}) and will be provided in the full version.
\end{proof}

\subsection{Consequences}

In this section we list some consequences of Theorem \ref{thm:closureprepost} and Theorem \ref{thm:main}.

In \cite{DBLP:journals/dc/AngluinAER07}, Angluin \textit{et al.} showed that IO protocols can compute exactly the counting predicates, i.e., the predicates that can be expressed by counting constraints. This is also a consequence of the proof of Theorem \ref{thm:main}. Moreover, our results allow us to go further, and provide a bound on the number of states required to compute a predicate.

\begin{corollary} 
 IO population protocols compute exactly the counting predicates, i.e., the predicates corresponding to counting constraints.  \end{corollary}
\begin{proof}
	Let $\PP$ be a well-specified IO protocol. The sets $\mathcal{I} \cap \overline{\pre^*(\overline{\pre^*(\mathcal{ST}_b)})}$ for \(b\in \{ 0,1 \}\) are the sets of initial configurations from which $\PP$ stabilizes to \(b=0,1\).  Theorem \ref{thm:closureprepost} shows that they are counting sets. 
\end{proof}

\begin{corollary} 
Let $\PP$ be an IO protocol computing a counting predicate $P(x_1,\ldots,x_k)$
of $U$-norm $u$ and $L$-norm $\ell$. 
Then there exists a constant $c$, independent of $\PP$, such that $\PP$ has at least 
$g \log \log (\max\{ u, \ell \})$ states, where $g$ denotes the inverse of the function $n \mapsto c \cdot (n^2 \log n)$.
\end{corollary}
\begin{proof}
	The set \(\mathcal{I} \cap \overline{\pre^*(\overline{\pre^*(\mathcal{ST}_1)})}\) describes the initial configurations that stabilize to $1$, i.e., the initial configurations for which the predicate computed by the protocol is true. 
	By Corollary~\ref{coro:closure} (using a reasoning similar to that of Theorem \ref{thm:main}), if $\PP$ has $n$ states, then  the $U$-norm and $L$-norm of \(\mathcal{I} \cap \overline{\pre^*(\overline{\pre^*(\mathcal{ST}_1)})}\) are bounded by the function $f(n)= 2^{2^{{\cal O}(n^2 \log n)}}$. Therefore, for a certain constant $c$, $\log \log \max\{ u, \ell \} \leq c \cdot (n^2 \log n)$ and the number of states of a protocol computing a predicate of $U$-norm $u$ and $L$-norm $\ell$ is at least $g \log \log (\max\{ u, \ell \})$, where $g(x)$ is the inverse function of $x \mapsto c \cdot (x^2 \log x)$. 
\end{proof}

Finally, we can show that the correctness problem for IO protocols is also in \EXPSPACE.

\begin{corollary}
Let $\PP$ be an IO population protocol with $n$ states and $k$ input states, and let $P(x_1, \ldots, x_k)$
be a counting predicate, expressed as a CoNF-constraint. The correctness problem for $\PP$ and $P$, i.e., the problem of deciding if $\PP$ computes $P$,  is in \EXPSPACE.
\end{corollary}
\begin{proof}[Proof Sketch.]
	We give a nondeterministic, exponential space algorithm for the complement of the correctness problem. The algorithm guesses nondeterministically a minterm of $\mathcal{I} \cap \overline{\pre^*(\overline{\pre^*(\mathcal{ST}_1)})}$, and checks
	that $\mathcal{I} \cap \overline{\pre^*(\overline{\pre^*(\mathcal{ST}_1)})}$ contains a configuration that does not satisfy $P$. 
	The algorithm does a similar check for \(\mathcal{ST}_0\) and a configuration that does satisfy \(P\).
	The minterm can be constructed in exponential space by Theorem \ref{thm:main}, and the check whether a minterm implies a CoNF-constraint can be done in polynomial time. 
\end{proof}

 \appendix
 \section{A PSPACE Lower Bound}\label{sec:lowerbound}
 
 We show that the well specification problem for IO protocols is \PSPACE-hard by reduction from the acceptance problem for 
 deterministic Turing machines running in linear space \cite{Papadimitriou}.

 Let $(Q, \Sigma, \Gamma, \delta, q_{\mathrm{init}}, q_{\mathrm{acc}}, q_{\mathrm{rej}})$ be a deterministic linear space Turing machine that uses exactly
 $n$ tape cells on an input of size $n$.
 As usual, $(q,a, q', a', d)\in\delta$ means that if the machine is currently in state $q$, 
 and its head sees symbol $a$ on the tape, 
 then the machine can move to state $q'$ while overwriting the letter $a$ to $a'$, 
 and moving the head in the direction $d\in \set{L, R}$.
 The machine accepts if it reaches $q_{\mathrm{acc}}\in Q$ and rejects if it reaches $q_{\mathrm{rej}}$.
 We assume w.l.o.g.\ that the machine does not leave the states $q_{\mathrm{acc}}$ or $q_{\mathrm{rej}}$
 once they are reached, and that either of these states is reached on every execution, and that the machine
 does not attempt to ``fall off'' the tape by moving left from the leftmost tape cell or right from the rightmost tape cell.
 A configuration of the machine consists of the current state $q\in Q$, the head position $i\in\set{1,\ldots,n}$, and
 the current tape contents $\Gamma^n$.
 
 We define a protocol scheme \(\mathcal{A}\) that simulates the behaviour of the machine; the protocol
 will be well-specified if and only if the machine does not accept.
 The protocol will use agents to track the configuration of the Turing machine.
 It will use an additional agent to guess and execute transitions.
 
 We define the possible states of the protocol.
 Let $\mathsf{Head} = \set{1,\ldots, n}$ be states used to track the head position.
 The protocol has the following set of states.
 \begin{itemize}
 \item \textbf{[Configuration states]} $Q \cup \mathsf{Head} \cup (\Gamma\times \mathsf{Head})$.
 An agent in state $q \in Q$ models that the machine is in state $q$, 
 an agent in $c\in \mathsf{Head}$ models that the head is at position $c$, 
 and an agent at $(a,c) \in \Gamma \times \mathsf{Head}$ models that cell $c$ contains letter $a$. 
 Of course, a population need not model a state of the machine consistently, because, e.g., 
 multiple participants may be in different states in $Q$. 
 The protocol will have rules to detect inconsistencies and move all agents to a special state.
 \item \textbf{[Transition states]} $\set{ t, \tuple{t\mid i}, \tuple{t\mid i,a}, \tuple{t\mid i,a, 1}, \tuple{t\mid i,a,2}\mid t\in\delta}\cup\set{\mathsf{start}}$.
 These states are used to simulate the execution of the transition of the machine,
 An agent in state $t$ denotes the simulation of the Turing machine by the protocol will execute the transition $t$.
 The other states are book-keeping states to keep track of intermediate steps in the simulation.
 \item \textbf{[Zombie]} A distinguished \emph{zombie state} $\mathsf{zombie}$.
 \end{itemize}
 All states have output \(1\) except \(q_{\mathrm{acc}}\) which outputs \(0\).
 
 To begin with, we place agents in $q_{init}$, $1$, and $(x_1, 1), \ldots, (x_n, n)$ to encode the start configuration
 of the Turing machine (with input $x_1\ldots x_n$), and agents in $\mathsf{start}$.
 A configuration of the protocol is said to be \emph{good for simulation} if there is 
 exactly one agent in any state in $Q$, exactly one agent in any state in $\mathsf{Head}$,
 exactly one agent in a state $(a,i) \in \Gamma\times\mathsf{Head}$ for each $i\in \mathsf{Head}$,
 and exactly one agent in any transition state, and no agent in the zombie state.
 We shall define transitions of the protocol that ensure that if the protocol is started from an
 initial configuration that is good for simulation, then we can simulate the behavior of the Turing
 machine for some number of steps and remain in a configuration that is good for simulation.
 Also, we shall add rules that if the protocol is started in a configuration that is not good for simulation,
 then eventually all agents enter the zombie state.

 With these invariants, we shall ensure that if the Turing machine accepts, then there is a run of the protocol
 starting from a good for simulation configuration which can reach a dissensus state.
 However, if the Turing machine does not accept, then for all input configurations, all reachable states are
 stable $1$-consensus states.
 Thus, the IO protocol will be ill-specified if{}f the Turing machine accepts.
 Next, we describe the transitions of the protocol.
 
 First, we show how configuration that are not good for simulation can be detected and how all states can become zombies
 in that case.
 The idea is that if any agent meets an agent in the zombie state, it converts its own state to zombie as well.
 By fairness, if there is any agent in a zombie state, then eventually all agents become zombies.
 Now, if a configuration is not good for simulation, by fairness, eventually two agents who together violate
 the good-for-simulation property must meet. At that point, we convert one of them to a zombie. 
 
 With this discussion, we only focus on initial configurations which are good for simulation.
 In the following, we use fairness to ensure that the sequence of transitions described below will be eventually executed.
 \subparagraph*{(Step 1)} Simulation always starts with an agent in state $\mathsf{start}$ observing an (unique!) agent in the state $q\in Q$ and guessing a transition in $\delta$.
 That is, the agent in state $\mathsf{start}$ observes $q$ and updates its state to a transition in $\delta$ with source $q$.
 Let us say the transition is $t \equiv (q, a, q', a', L) \in \delta$ (the case of moving right is analogous).
 
 \subparagraph*{(Step 2)}%
 \label{par:_step_2_}
 The (unique) agent in state $t$ observes the (unique) agent in state $i\in\mathsf{Head}$ and updates its state to $\tuple{t\mid i}$.
 This encodes the current head position along with the guessed transition.
 
 \subparagraph*{(Step 3)}%
 \label{par:_step_3_}
 The agent $\tuple{t\mid i}$ observes the agent $(a,i)$ and updates its state to $\tuple{t\mid i,a}$.
 At this point, this agent has ensured that the transition $t$ can fire from the current configuration.
 However, if the agent meets an agent $(b, i)$ for $b\neq a$, it goes back to state $\mathsf{start}$, because the transition
 was guessed incorrectly.
 
 \subparagraph*{(Step 4)}%
 \label{par:_step_4_}
 Now, the agent in state $q$ observes $\tuple{t\mid i,a}$ and updates its state to $q'$.
 Since we started from a good for simulation configuration, the new configuration is still good for simulation and the unique
 agent encoding the state of the machine is in state $q'$.
 
 \subparagraph*{(Step 5)}%
 \label{par:_step_5_}
 The agent $\tuple{t\mid i,a}$ observes $q'$ and updates its state to $\tuple{t\mid i,a,1}$, where the ``$1$'' encodes that
 the state transition has been made.
 It is possible that $q = q'$ and step 4 is omitted. However, the effect on the configuration is the same in this case.
 
 \subparagraph*{(Step 6)}%
 \label{par:_step_6_}
  Now we update the tape cell.
 The agent $(a, i)$ observes $\tuple{t\mid i,a,1}$ and updates itself to $(a',i)$.
 
 \subparagraph*{(Step 7)}%
 \label{par:_step_7_}
 The agent $\tuple{t\mid i,a,1}$ observes $(a', i)$ and moves to $\tuple{t\mid i,a,2}$: the tape cell has been correctly
 updated.
 If $a = a'$, step 6 may be omitted, but the combined effect of these two steps keeps the simulation valid.
 
 \subparagraph*{(Step 8)}%
 \label{par:_step_8_}
 Now we update the head position.
 The agent in state $i\in \mathsf{Head}$ observes the agent in $\tuple{t\mid i,a,2}$ and updates itself to $i-1$.
 (Or to $i+1$ if the transition moves right.)
 
 \subparagraph*{(Step 9)}%
 \label{par:_step_9_}
  Finally, $\tuple{t\mid i,a,2}$ observes $i-1$ and moves to $\mathsf{start}$.
 
 In sum, all these transitions correctly encode one step of the Turing machine and keeps the protocol in a good for simulation state.
 Soundness is easy to show: there is a scheduler that always chooses the correct agents and thus simulates the Turing machine.
 Thus, if the machine accepts, we get to a good for simulation state with an agent in $q_{\mathit{acc}}$.
 From this point, no transition fires because there is no transition with this state as source.
 But this configuration has one agent with output $0$ and others with output $1$, and is a dissensus. 
 On the other hand, by induction, we can show completeness: the simulation either maintains a reachable configuration or a partially
 executed reachable configuration.
 
 An inspection of the interaction rules reveals it is an immediate observation protocol.
 Together, we have shown that the IO protocol can reach a dissensus if{}f the Turing machine accepts.

\end{document}